\documentclass[preprint,12pt]{elsarticle}
\usepackage{natbib}
\biboptions{sort&compress}
\usepackage{multirow}
\usepackage{makecell}
\usepackage[colorlinks]{hyperref}
\usepackage{amssymb}
\usepackage{amsmath}
\usepackage{amsthm}
\newtheorem{theorem}{Theorem}[section]
\newtheorem{lemma}[theorem]{Lemma}

\newtheorem{proposition}[theorem]{Proposition}
\newdefinition{defind}[theorem]{Definition}
\newdefinition{remark}[theorem]{Remark}
\newdefinition{example}[theorem]{Example}

\begin{document}

\begin{frontmatter}

\title{A Characterization of MDS Symbol-pair Codes over Two Types of Alphabets }

\author[a]{Xilin Tang}
\ead{xltang@scut.edu.cn}
\author[a]{Weixian Li}
\ead{201910105912@mail.scut.edu.cn}
\author[b,c]{Wei Zhao \corref{cor1}}
\ead{zhaowei@cuhk.edu.cn}

\cortext[cor1]{Corresponding author}

\address[a]{Department of Mathematics, South China University of Technology, Guangzhou,
Guangdong, 510640, PR China}
\address[b]{The Chinese University of Hong Kong, Shenzhen,
Guangdong, 518172, PR China}
\address[c]{University of Science and Technology of China, Anhui, 230026, PR China}

\begin{abstract}
Symbol-pair codes are block codes with symbol-pair metrics designed to protect against pair-errors that may occur in high-density data storage systems. MDS symbol-pair codes are optimal in the sense that it can attain the highest pair-error correctability  within the same  code length and code size. Constructing MDS symbol-pair codes is one of the main topics in symbol-pair codes. In this paper, we characterize the symbol-pair distances of some constacyclic codes of arbitrary lengths over finite fields and a class of finite chain rings. Using the characterization of symbol-pair distance, we present several classes of MDS symbol-pair constacyclic codes and show that there is no other MDS symbol-pair code among the class of constacyclic codes except for what we present. Moreover, some of these MDS symbol-pair constacyclic codes over the finite chain rings cannot be obtained by previous work.
\end{abstract}

\begin{keyword}
Symbol-pair codes \sep MDS codes \sep constacyclic codes
\end{keyword}

\end{frontmatter}

\section{Introduction}
Modern high-density data storage systems may not read the transmitted information individually as classic information transmission due to physical limitations. Motivated by this fact, Cassuto and Blaum \cite{Cassuto2011} developed symbol-pair code over \textit{symbol-pair read channel} whose outputs are overlapping pairs of symbols. The efficient decoding algorithms for cyclic codes over symbol-pair read channels are shown in \cite{Yaakobi2012,Takita2015,Morii2016}.

Let $\Xi$ be an alphabet of $q$ elements with $q\geq 2$. A \textit{code} $\mathcal{C}$ over $\Xi$ of length $n$ is a subset of $\Xi^n$. The elements in $\mathcal{C}$ are called \textit{codewords}. We use the bold letter to denote a vector in the sequel. Let $\boldsymbol{x}=\left(x_{0}, x_{1}, \ldots, x_{n-1}\right)$, $\boldsymbol{y}=\left(y_{0}, y_{1}, \ldots, y_{n-1}\right)$ be vectors in $\Xi^{n}$. A vector $\boldsymbol{x}$ transmitted in the symbol-pair read channel is read as
\[\pi\left(\boldsymbol{x}\right)=\left(\left(x_{0}, x_{1}\right),\left(x_{1}, x_{2}\right), \ldots,\left(x_{n-1}, x_{0}\right)\right).\]
We call $\pi(\boldsymbol{x})$ as a \emph{symbol-pair vector} of $\boldsymbol{x}$. The \textit{symbol-pair distance} between $\boldsymbol{x}$ and $\boldsymbol{y}$ is defined as the Hamming distance between $\pi(\boldsymbol{x})$ and $\pi(\boldsymbol{y})$, $i.e.$,
$$\operatorname{d_{sp}}(\boldsymbol{x}, \boldsymbol{y})=\operatorname{d_{H}}(\pi(\boldsymbol{x}), \pi(\boldsymbol{y}))=\left|\left\{i:\left(x_{i}, x_{i+1}\right) \neq\left(y_{i}, y_{i+1}\right)\right\}\right|.$$
The \textit{(minimum) symbol-pair distance} of $\mathcal{C}$ is defined as
$$\operatorname{d_{sp}}(\mathcal{C})=\min\{\operatorname{d_{sp}}(\boldsymbol{x},\boldsymbol{y})\mid \boldsymbol{x},\boldsymbol{y}\in \mathcal{C}~\text{and}~\boldsymbol{x}\neq \boldsymbol{y}\}.$$
For a code $\mathcal{C}$ of length $n$ over $\Xi$ with symbol-pair distance $\operatorname{d_{sp}}$, the upper bound on the code size of $\mathcal{C}$, called \textit{Singleton bound for symbol-pair codes} \cite{Chee2013}, is
\begin{equation}\label{equa1}
|\mathcal{C}| \leq q^{n-\operatorname{d_{sp}}+2}.
\end{equation}
A symbol-pair code whose parameters satisfy \eqref{equa1} with equality is called \textit{maximum distance separable} (MDS). According to \eqref{equa1}, MDS symbol-pair codes possess the largest symbol-pair distance under the same code length and code size, which indicates that MDS symbol-pair codes are a class of optimal symbol-pair codes that can have high pair error-correcting capability since the symbol-pair distance is a tool to measure the pair error-correcting capability of the codes.

Constructing MDS symbol-pair codes is meaningful both in theoretical and practical.
The research on constructing MDS symbol-pair codes is active in recent years \cite{Li2016,Chen2017,Kai2018,Dinh2018,Dinh2019b,Dinh2020}. Many MDS symbol-pair codes are obtained by analyzing the generator polynomials of constacyclic codes. See \cite{Li2016,Chen2017,Kai2018,Zhao2020} for example. In \cite{Dinh2018}, Dinh \emph{et al.} characterize the symbol-pair distances of all constacyclic codes of length $p^s$ over $\mathbb{F}_{p^m}$ and obtain all the MDS symbol-pair codes of prime power lengths. These results are generalized in two different directions. One of the directions is to construct MDS symbol-pair constacyclic codes over different alphabets such as the finite chain ring $\mathbb{F}_{p^{m}}+u \mathbb{F}_{p^{m}}$ (\cite{Dinh2018a,Dinh2019b}). The other direction is extending the code length of constacyclic codes to some other special code lengths such as $2 p^s$ (\cite{Dinh2019a,Dinh2020}).

Let $\lambda$ be a nonzero element in $\mathbb{F}_{p^{m}}$ and $n$ be a positive integer coprime to $p$. Due to the complicated irreducible factorization of $x^n-\lambda$ in $\mathbb{F}_{p^{m}}[x]$, the algebraic structure of $\lambda$-constacyclic codes of length $np^{s}$ over $\mathbb{F}_{p^{m}}$ are not obtained completely. Therefore it is difficult to analyze the symbol-pair distance of constacyclic codes of length $n p^s$. In \cite{Ozadam2009}, the authors discussed the structure of a special class of constacyclic codes of length $n p^s$ over $\mathbb{F}_{p^{m}}$. This work inspires us to analyze the symbol-pair distances of these constacyclic codes. Our motivation is to characterize the MDS symbol-pair codes among the larger class of constacyclic codes and obtain new MDS symbol-pair codes with more flexible parameters.

In this paper, we consider some constacyclic codes of length $np^{s}$ over two different alphabets, which are finite fields $\mathbb{F}_{p^m}$ and finite chain rings $\mathbb{F}_{p^{m}}+u \mathbb{F}_{p^{m}}$, where $p$ is a prime, $m$ is a positive integer and $u^2=0$. Let $\alpha_0$ be a nonzero element in $\mathbb{F}_{p^m}$ such that $x^n-\alpha_0$ is irreducible over $\mathbb{F}_{p^m}$. Denote $\alpha=\alpha_0^{p^s}$. Let $\beta$ be an element in $\mathbb{F}_{p^m}$. We completely characterize the symbol-pair distances of $\alpha$-constacyclic codes over $\mathbb{F}_{p^m}$ and $(\alpha+u\beta)$-constacyclic codes over $\mathbb{F}_{p^{m}}+u \mathbb{F}_{p^{m}}$. We present several classes of MDS symbol-pair constacyclic codes in Table \ref{tab1} and Table \ref{tab2}. Some of these codes are obtained in previous work and we remark the references in the tables. Some of these codes are obtained in this paper.

\begin{table}[h]
\centering
\begin{tabular}{|c|c|c|c|c|}
\hline
\makecell{Generator \\ Polymial} & Dimension & \makecell{Pair \\Distance} & Remark & Ref. \\
\hline
$ x-\alpha_0 $ & $p^s-1$ & 3 & ~ & \cite{Dinh2018}\\
\hline
$ (x-\alpha_0)^{2} $ & $p^s-2$ & 4 & ~ & \cite{Dinh2018}\\
\hline
$ (x-\alpha_0)^{4} $ & 5 & 6 & \makecell[c]{$p=3$ \\ $s=2$} & \cite{Dinh2018}\\
\hline
$ (x-\alpha_0)^{k} $ & $p-k$ & $k+2$ & \makecell{$s=1$ \\ $1 \leq k \leq p-2$} & \cite{Dinh2018}\\
\hline
$ (x-\alpha_0)^{p^s-2} $ & 2 & $p^s$ & ~ & \cite{Dinh2018}\\
\hline
$ x^2-\alpha_0 $ & $2p^s-2$ & 4 & ~ & \cite{Dinh2019a}\\
\hline
$ (x^2-\alpha_0)^{k} $ & $2p-2k$ & $2k+2$ & \makecell{$s=1$ \\ $1 \leq k \leq p-2$} & \cite{Dinh2019a}\\
\hline
$ (x^2-\alpha_0)^{p^s-1} $ & 2 & $2p^s$ & ~ & \cite{Dinh2019a}\\
\hline
\end{tabular}
\caption{MDS symbol-pair $\alpha$-constacyclic codes of length $np^s$ over $\mathbb{F}_{p^m}$}
\label{tab1}
\end{table}

\begin{table}[h]
\centering
\begin{tabular}{|c|c|c|c|c|}
\hline
\makecell{Generator\\ Polynomial} & Size & \makecell{Pair \\Distance}  & Remark & Ref. \\
\hline
$(x-\alpha_0)+u b(x)$ & $p^{2m(p^s-1)}$ & 3 & ~ & \cite{Sharma2019}\\
\hline
\makecell{$(x-\alpha_0)^{2} +$ \\ $u (x-\alpha_0)^{k} b(x) $} & $p^{2m(p^s-2)}$ & 4 & \makecell{$s \geq 2$ \\ $k=0,1$} & \cite{Sharma2019}\\
\hline
\makecell{$(x-\alpha_0)^{4} +$ \\ $u (x-\alpha_0)^{k} b(x) $} & $p^{10m}$ & 6 & \makecell[c]{$p=3$ \\ $s=2$ \\ $0 \leq k \leq 3$} & \cite{Sharma2019}\\
\hline
\makecell{$(x-\alpha_0)^{j}  +$ \\ $u (x-\alpha_0)^{k} b(x) $} & $p^{2m(p-k)}$ & $j+2$  & \makecell{$s=1$ \\ $1 \leq j \leq p-2$ \\ $max\{0,2j-p\}$ \\ $\leq k < j$} & \cite{Sharma2019}\\
\hline
\makecell{$(x-\alpha_0)^{p^s-2} +$ \\ $ u (x-\alpha_0)^{k} b(x) $} & $p^{4m}$ & $p^s$ & \makecell{$k=p^s-4, p^s-3$} & \cite{Sharma2019} \\
\hline
\makecell{$(x^2 -\alpha_0)+$ \\ $u b(x)$} & $p^{4m(p^s-1)}$ & 4 & ~ & Thm. \ref{th4.6} \\
\hline
\makecell{$(x^2 -\alpha_0)^{j}  +$ \\ $ u (x^2 -\alpha_0)^{k} b(x) $} & $p^{4m(p-k)}$ & $2j+2$ & \makecell{$s=1$ \\ $1 \leq j \leq p-2$ \\ $max\{0,2j-p\}$ \\ $\leq k < j$} & Thm. \ref{th4.6} \\
\hline
\makecell{$(x^2 -\alpha_0)^{p^s-1} +$ \\ $ u (x^2 - \alpha_0)^{p^s-2} b(x) $} & $p^{4m}$ & $2p^s$ & ~ & Thm. \ref{th4.6} \\
\hline
\end{tabular}
\caption{MDS symbol-pair $\alpha$-constacyclic codes of length $np^s$ over $\mathbb{F}_{p^m}+u \mathbb{F}_{p^m}$, where b(x) is either zero or a unit in $\mathbb{F}_{p^m}[x] / \langle x^{np^s} - \alpha \rangle$}
\label{tab2}
\end{table}

The codes in Table \ref{tab1} are MDS symbol-pair $\alpha$-constacyclic codes over $\mathbb{F}_{p^m}$. For any positive integer $n$, we prove that there is no other MDS symbol-pair $\alpha$-constacyclic code of length $np^s$ except for the codes in Table \ref{tab1}.

The codes in Table \ref{tab2} are MDS symbol-pair $(\alpha+u\beta)$-constacyclic codes over $\mathbb{F}_{p^m}+u\mathbb{F}_{p^m}$. Notice that the codes considered in \cite{Sharma2019} is a subcase of the codes we considered in this paper which confine $n=1$.
In \cite{Dinh2019b}, Dinh \emph{et al.} gave two classes of MDS symbol-pair codes with parameters $(2^s, 2^{m(2^{s-1}+4)}, 3 \cdot 2^{s-2})$ and $(3^s, 3^{m(2 \cdot 3^{s-1}+4)}, 2 \cdot 3^{s-1})$, but these two classes are actually not MDS symbol-pair codes.
We will give a detailed analysis of these two classes of codes in section \ref{sec4}. Besides, we also obtain three new classes of MDS symbol-pair $(\alpha+u\beta)$-constacyclic codes.  Moreover, we prove that there is no other MDS symbol-pair $(\alpha+u\beta)$-constacyclic codes of length $np^s$ except for the codes we present in Table \ref{tab2}.

The remaining of this paper is organized as follows. In Section \ref{sec2}, we present some preliminaries and notations. In Section \ref{sec3}, we give the symbol-pair distances of all $\alpha$-constacyclic codes of length $n p^s$ over $\mathbb{F}_{p^m}$ and show all the MDS symbol-pair $\alpha$-constacyclic codes of length $n p^s$ among these codes. In Section \ref{sec4}, we determine the symbol-pair distances of some of $(\alpha+u\beta)$-constacyclic codes of length $n p^s$ over $\mathbb{F}_{p^{m}}+u \mathbb{F}_{p^{m}}$ and exhibit all the MDS symbol-pair codes among these codes.

\section{Preliminaries}\label{sec2}

In this section, we give some notations and results that will be used in the sequel.

Let $R$ be a finite commutative ring with identity. A code $\mathcal{C}$ over $R$ is called \textit{linear} if $\mathcal{C}$ is a submodule of $R^n$. The \textit{symbol-pair weight} of a vector $\boldsymbol{x}$ in $R^n$ is the symbol-pair distance between $\boldsymbol{x}$ and the all-zero vector $\boldsymbol{0}$ of $R^n$, denoted by $\operatorname{wt_{sp}}(\boldsymbol{x})$.  The symbol-pair distance of a linear code is equal to the minimum symbol-pair weight of nonzero codewords of the linear code. For a unit $\lambda$ of $R$, the $\lambda$-constacyclic shift $\tau_{\lambda}$ on $R^{n}$ is defined by:
\[
\tau_{\lambda}\left(x_{0}, x_{1}, \ldots, x_{n-1}\right)=\left(\lambda x_{n-1}, x_{0}, x_{1}, \ldots, x_{n-2}\right).
\]
A linear code $\mathcal{C}$ is said to be $\lambda$-\textit{constacyclic} if $\tau_{\lambda}(\mathcal{C})=\mathcal{C}$. 
Each codeword $\boldsymbol{c}=\left(c_{0}, c_{1}, \ldots, c_{n-1}\right)$ in $\mathcal{C}$ is customarily identified with its polynomial representation $c(x)=c_{0}+c_{1} x+\cdots+c_{n-1} x^{n-1}$ in ${R[x]}/{\left\langle x^{n}-\lambda\right\rangle}$. In the ring ${R[x]}/{\left\langle x^{n}-\lambda\right\rangle}$, $x c(x)$ corresponds to performing a $\lambda$-constacyclic shift on $\boldsymbol{c}$. The following theorem shows the algebraic property of constacyclic codes.

\begin{proposition}\label{prop2.2}\cite{Huffman2003}
A linear code $\mathcal{C}$ of length $n$ over $R$ is a $\lambda$-constacyclic code if and only if $\mathcal{C}$ is an ideal of the quotient ring ${R[x]}/{\left\langle x^{n}\!-\!\lambda\right\rangle}$.
\end{proposition}

The ideal of ${R[x]}/{\left\langle x^{n}-\lambda\right\rangle}$ is generated by a factor of $x^{n}-\lambda$. Let $\lambda$ be a unit of $R$. In this paper, we mainly consider the constacyclic codes of length $np^s$, i.e., an ideal of the residue ring ${R[x]}/{\left\langle x^{np^s}-\lambda\right\rangle}$. If $R$ is a Frobenius ring, one can find a unit $\lambda_0$ such that $\lambda_0^{p^s}=\lambda$. Therefore, we have $x^{np^s}-\lambda=(x^n-\lambda_0)^{p^s}$. In this paper, we make an assumption that $x^n-\lambda_0$ is irreducible in $R[x]$. The following result shows the irreducibility of binomials over finite fields.

\begin{proposition}\label{prop2.3}\cite{RudolfLidl2008}
Let $n \geq 2$ be an integer and $\lambda \in \mathbb{F}_{q}^{*}$. Then the binomial $x^{n}-\lambda$ is irreducible in $\mathbb{F}_{q}[x]$ if and only if the following two conditions are satisfied:
\begin{enumerate}[(i)]
\item each prime factor of $n$ divides the order $e$ of $\lambda$ in $\mathbb{F}_{q}^{*},$ but not $\frac{q-1}{e}$;
\item $q \equiv 1(\bmod\; 4)$ if $n \equiv 0(\bmod\; 4)$.
\end{enumerate}
\end{proposition}

According to Proposition 2.2, if $n$ satisfy the condition (ii) of Proposition \ref{prop2.3} and all the prime factors of $n$ divide $q-1$, there exists $\lambda \in \mathbb{F}_{q}^{*}$ such that  $x^{n}-\lambda$ is irreducible over $\mathbb{F}_{q}$.

\subsection{Constacyclic Codes over \texorpdfstring{$\mathbb{F}_{p^m}$}{TEXT}}

Let $\alpha$ be a nonzero element in $\mathbb{F}_{p^m}$. We present some results of $\alpha$-constacyclic codes of length $n p^s$ over $\mathbb{F}_{p^m}$ in this subsection. Denote by $\mathcal{F}$ the quotient ring ${\mathbb{F}_{p^{m}}[x]}/{\left\langle x^{np^{s}}-\alpha\right\rangle}$. The structures and the minimum Hamming distances of $\alpha$-constacyclic codes are given in the following theorem.

\begin{theorem}\label{th_H_F}\cite[Theorem 3.6]{Ozadam2009}
Let $\mathbb{F}_{p^m}$ be a finite field and $n$ be a positive integer with $\operatorname{gcd}(n,p)=1$. Suppose that $x^n-\alpha_0$ is irreducible over $\mathbb{F}_{p^m}$ for $\alpha_0 \in \mathbb{F}_{p^m}^{*}$ and $\alpha={\alpha_0}^{p^s}$. Then the $\alpha$-constacyclic codes of length $n p^s$ over $\mathbb{F}_{p^m}$ are of the form $\mathcal{C}_i=\langle\left(x^{n}-\alpha_{0}\right)^{i}\rangle$, where $0 \leq i \leq p^s$. And the minimum Hamming distance of $\mathcal{C}_i$ is given by
$$d_H(\mathcal{C}_i)\!\!=\!\!\left\{\!\!\!\begin{array}{ll}
1, & \text { if } i=0, \\
(\theta+2) p^{k}, & \text { if } p^{s}\!-\!p^{s-k}\!+\!\theta p^{s-k-1}\!+\!1 \!\leq\! i \!\leq\!  p^{s}\!-\!p^{s-k}\!+\!(\theta\!+\!1) p^{s-k-1}, \\
&  \text { where } 0 \leq \theta \leq p-2 \text { and } 0 \leq k \leq s-1, \\
0 & \text { if } i=p^{s}.
\end{array}\right.$$
\end{theorem}

For simplicity, we use the notation $\mathcal{C}_i$ to denote the $\alpha$-constacyclic codes of length $np^s$ with generator polynomial $(x^n-\alpha_{0})^{i}$, where $0\leq i\leq p^s$. The following lemma shows a formula to compute the Hamming weight of the codeword $(x^n-\alpha_{0})^{i}$ in $\mathcal{C}_{i}$.

\begin{lemma}\label{lem_H_bi}\cite[Lemma 1]{Massey1973}
For any nonnegative integer $i<p^{s},$ let $i=i_{s-1} p^{s-1}+\cdots+i_{1} p+i_{0}$, where $0 \leq i_{0}, i_{1}, \ldots, i_{s-1} \leq p-1$, which means that $(i_{s-1},\ldots,i_0)$ is the $p$-adic expansion of $i$. Then
$$\operatorname{wt_H}((x^n-\alpha_0)^i)=\prod_{j=0}^{s-1}(i_j +1).$$
\end{lemma}

The following lemma shows the relationship between the symbol-pair distance and the Hamming distance.

\begin{lemma}\label{th_H_S}\cite[Theorem 2]{Cassuto2011}
For two codewords $\boldsymbol{x},\boldsymbol{y}$ in a code $\mathcal{C}$ of length $n$ with $0 < \operatorname{d_{H}}(\boldsymbol{x},\boldsymbol{y}) < n$, define the set $S_H=\{j \mid x_j \neq y_j\}$. Let $S_H=\cup_{l=1}^L B_l$ be a minimal partition of the set $S_H$ to subsets of consecutive indices(indices may wrap around modulo $n$). Then
$$\operatorname{d_{sp}}(\boldsymbol{x},\boldsymbol{y})=\operatorname{d_{H}}(\boldsymbol{x},\boldsymbol{y})+L.$$
\end{lemma}

To calculate the symbol-pair distances, we will use the concept of the \textit{coefficient weight} of polynomials, which was first proposed in \cite{Dinh2007}. For a polynomial $f(x)\!=\!a_{n} x^{n}+\cdots+a_{1} x+a_{0}$ of degree $n$, the coefficient weight of $f$, which denoted by $\operatorname{cw}(f)$, is
$$\operatorname{cw}(f)=\left\{\begin{array}{ll}
0, & \text { if } f \text { is a monomial } \\
\min \left\{|i-j|: a_{i} \neq 0, a_{j} \neq 0, i \neq j\right\}, & \text { otherwise. }
\end{array}\right.$$
Intuitively, $\operatorname{cw}(f)$ is the smallest distance among exponents of nonzero terms of $f(x)$. It is shown in \cite{Dinh2018} that if $0 \leq \operatorname{deg}(g(x)) \leq \operatorname{cw}(f(x))-2$ and $\operatorname{deg}(f(x))+\operatorname{deg}(g(x)) \leq n-2$, then
\begin{equation}\label{equ2.2}
\operatorname{wt_{sp}}(f(x) g(x))=\operatorname{wt_{H}}(f(x)) \cdot \operatorname{wt_{sp}}(g(x)).
\end{equation}

\subsection{Constacyclic Codes over \texorpdfstring{$\mathbb{F}_{p^{m}}+u\mathbb{F}_{p^{m}}$}{TEXT}}
\label{sub 2}

Let $\alpha+u\beta$ be a unit in $\mathbb{F}_{p^{m}}+u\mathbb{F}_{p^{m}}$, i.e., $\alpha\neq 0$. This subsection gives the structures of $(\alpha+u\beta)$-constacyclic codes of length $n p^s$ over $\mathbb{F}_{p^{m}}+u\mathbb{F}_{p^{m}}$. Denote by $\mathcal{R}$ the quotient ring $(\mathbb{F}_{p^{m}}+u\mathbb{F}_{p^{m}})[x] /\left\langle x^{n p^{s}}-\alpha-u \beta\right\rangle$.

Note that the structures of the ideals of $\mathcal{R}$ are quite different when the value of $\beta$ is equal to zero or not. The following lemma shows that all the ideals of $\mathcal{R}$ are principal ideals in the case of $\beta\neq 0$.
\begin{lemma}\label{th_R_1}\cite[Theorem 3.3]{Zhao2018}
Let $x^n-\alpha_0$ be an irreducible polynomial in $\mathbb{F}_{p^m}[x]$, $\alpha=\alpha_{0}^{p^s}$, and $\beta$ be a nonzero element in $\mathbb{F}_{p^m}$. Then the ring $\mathcal{R}=(\mathbb{F}_{p^{m}}+u\mathbb{F}_{p^{m}})[x] /\left\langle x^{n p^{s}}-\alpha-u \beta\right\rangle$ is a chain ring whose ideal chain is as follows
$$\mathcal{R}=\langle 1\rangle \supsetneq \langle x^{n}-\alpha_{0} \rangle \supsetneq \cdots \supsetneq \langle\left(x^{n}-\alpha_{0}\right)^{2 p^{s}-1} \rangle \supsetneq \langle\left(x^{n}-\alpha_{0}\right)^{2 p^{s}}\rangle=\langle 0\rangle .$$
In other words, $(\alpha+u\beta)$-constacyclic codes of length $np^s$ over $\mathbb{F}_{p^{m}}+u\mathbb{F}_{p^{m}}$ are precisely the ideals $\mathcal{D}_i=\langle\left(x^{n}-\alpha_{0}\right)^{i}\rangle$ of $\mathcal{R}$, where $0 \leq i \leq 2p^s$. The number of codewords of $(\alpha+u\beta)$-constacyclic code $\mathcal{D}_i$ is $p^{mn(2p^s-i)}$. In particular, $\langle\left(x^{n}-\alpha_{0}\right)^{p^{s}}\rangle=\langle u\rangle$.
\end{lemma}

For the remaining case of $\beta = 0$, there are three types of ideals in $\mathcal{R}$.

\begin{lemma}\label{th_R_0}\cite[Corollary 3.10]{Cao2015}
Let $x^n-\alpha_0$ be an irreducible polynomial in $\mathbb{F}_{p^m}[x]$, $\alpha=\alpha_{0}^{p^s}$, and $\mathcal{F}=\mathbb{F}_{p^m}[x] / \langle x^{n p^s}-\alpha \rangle$. Then all $\alpha$-constacyclic codes over $\mathbb{F}_{p^m} + u \mathbb{F}_{p^m}$ of length $n p^s$, $i.e.$ all ideals of the ring $\mathcal{F}+u\mathcal{F}$, are given by the following three types:

(I) $\mathcal{D}=\langle (x^n-\alpha_0)^k \rangle$, where $0 \leq k \leq p^s$, with $|\mathcal{D}|=p^{2mn(p^s-k)}$

(II) $\mathcal{D}=\langle\left(x^{n}-\alpha_{0}\right)^{j} b(x)+u\left(x^{n}-\alpha_{0}\right)^{k}\rangle$, where $0 \leq k \leq p^s-1$, $\left\lceil\frac{p^{s}+k}{2}\right\rceil \leq j \leq p^s-1$ and either $b(x)$ is 0 or $b(x)$ is a unit in $\mathcal{F}$, with $|\mathcal{D}|=p^{mn(p^s-k)}$.

(III) $\mathcal{D}=\langle\left(x^{n}-\alpha_{0}\right)^{j} b(x)+u\left(x^{n}-\alpha_{0}\right)^{k}, \left(x^{n}-\alpha_{0}\right)^{k+t} \rangle$, where $0 \leq k \leq p^s-2$, $1 \leq t \leq p^s-k-1$, $k + \left\lceil\frac{t}{2}\right\rceil \leq j \leq  k+t$, and either $b(x)$ is 0 or $b(x)$ is a unit in $\mathcal{F}$, with $|\mathcal{D}|=p^{mn(2p^s-2k-t)}$.

\end{lemma}

Note that $\mathbb{F}_{p^{m}}$ is a subfield of $\mathbb{F}_{p^{m}}+u\mathbb{F}_{p^{m}}$. We define the \textit{subfield subcode} of codes over $\mathbb{F}_{p^{m}}+u\mathbb{F}_{p^{m}}$ as the set of codewords whose components are in $\mathbb{F}_{p^m}$. We use the notation $\left.\mathcal{D}\right|_{F}$ to denote the subfield subcode of $\mathcal{D}$ and $\operatorname{d_{sp}}\left(\mathcal{D}_{F}\right)$ to denote the symbol-pair distance of $\left.\mathcal{D}\right|_{F}$. We can represent the polynomial $c(x)$ over $\mathbb{F}_{p^{m}}+u\mathbb{F}_{p^{m}}$ as $c(x)=a(x)+u b(x)$, where $a(x), b(x) \in \mathbb{F}_{p^{m}}[x]$. Observing that $c_{i}=a_{i}+u b_{i}=0$ if and only if $a_{i}=b_{i}=0$, where $c_i$, $a_i$ and $b_i$ are coefficients of $x^i$ in polynomials $c(x)$, $a(x)$ and $b(x)$, respectively. It follows that $\operatorname{wt_{sp}}(c(x)) \geq \max\left\{\operatorname{{wt}_{sp}}(a(x)), \operatorname{{wt}_{sp}}(b(x))\right\}$.

\section{MDS Symbol-pair codes over \texorpdfstring{$\mathbb{F}_{p^m}$}{TEXT}}\label{sec3}

\subsection{Symbol-pair distances of constacyclic codes}
Denote $\mathcal{C}_{i}=\langle (x^n-\alpha)^i\rangle$ as the $\alpha$-constacyclic codes of length $np^s$ over $\mathbb{F}_{p^{m}}$, where $0\leq i\leq p^s$.
In this subsection, we give a complete characterization of the symbol-pair distances of $\mathcal{C}_i$. The symbol-pair distances for the cases that $n=1$ and $n\geq 2$ are different, and we only give the analysis of the case for $n\geq 2$. For the related results of the case of $n=1$, we refer the readers to~\cite{Dinh2018}.

We discuss the symbol-pair distance of $\mathcal{C}_i$ depends on the value of $i$. For the trivial cases that $i=0$ and $i=p^s$, we have
\[\operatorname{d_{sp}}(\mathcal{C}_0)=\operatorname{d_{sp}}(\mathcal{F})=2 \] and
\[\operatorname{d_{sp}}(\mathcal{C}_{p^s})=\operatorname{d_{sp}}(\langle 0 \rangle)=0.\]

In order to consider the symbol-pair distances of $\mathcal{C}_{i}$ for $1\leq i\leq p^s-1$, we divide the set $\{ i \in \mathbb{N}, 1 \leq i \leq p^s-1\}$ into $s(p-1)$ parts, i.e.,
$$\begin{aligned}
&\cup_{\substack{0 \leq k \leq s-1 \\ 0 \leq \theta \leq p-2}}\{i \in \mathbb{N} , p^s-p^{s-k}+\theta p^{s-k-1} +1 \leq i \leq p^s-p^{s-k}+(\theta+1)p^{s-k-1}\}\\
&=\{ i \in \mathbb{N} , 1 \leq i \leq p^s-1\}.
\end{aligned}$$

Note that $\operatorname{d_{sp}}(\mathcal{C}_i)\leq \operatorname{d_{sp}}(\mathcal{C}_j)$ if $i\leq j$ since $\mathcal{C}_{i}\supseteq \mathcal{C}_{j}$. In order to determine the symbol-pair distances of $\mathcal{C}_i$ for $p^{s}-p^{s-k}+\theta p^{s-k-1}+1 \leq i \leq p^{s}-p^{s-k}+(\theta+1) p^{s-k-1}$, where $0 \leq k \leq s-1$ and $0 \leq \theta \leq p-2$, we consider an upper bound $U$ on the symbol-pair distance of $\mathcal{C}_{p^{s}-p^{s-k}+(\theta+1) p^{s-k-1}}$ and a lower bound $L$ on the symbol-pair distance of $\mathcal{C}_{p^{s}-p^{s-k}+\theta p^{s-k-1}+1}$. Observing that $\mathcal{C}_{p^{s}-p^{s-k}+\theta p^{s-k-1}+1}\supseteq \mathcal{C}_{p^{s}-p^{s-k}+(\theta+1) p^{s-k-1}}$, therefore,
$$L\leq \operatorname{d_{sp}}(\mathcal{C}_{p^{s}-p^{s-k}+\theta p^{s-k-1}+1})\leq \operatorname{d_{sp}}(\mathcal{C}_{p^{s}-p^{s-k}+(\theta+1) p^{s-k-1}})\leq U.$$
If $L=U$, then we obtain the symbol-pair distances of $\mathcal{C}_{i}$ for each $i$ belongs to the interval $[p^{s}-p^{s-k}+\theta p^{s-k-1}+1,p^{s}-p^{s-k}+(\theta+1) p^{s-k-1}]$.
The following lemma shows the corresponding upper bound.

\begin{lemma}\label{lem3.1}
Let $n, k, \theta$ be integers such that $n \geq 2$, $0 \leq k \leq s-1$ and $0 \leq \theta \leq p-2$. Then $\operatorname{d_{sp}}\left(\mathcal{C}_{p^{s}-p^{s-k}+(\theta+1) p^{s-k-1}}\right) \leq 2(\theta+2) p^{k}$.
\end{lemma}
\begin{proof}
By Lemma \ref{lem_H_bi}, we have
$$\operatorname{wt_H}(\left(x^n-\alpha_{0}\right)^{p^{s}-p^{s-k}+(\theta+1) p^{s-k-1}})= (\theta+2) p^k.$$
Since $$\operatorname{cw}( \left(x^n-\alpha_{0}\right)^{p^{s}-p^{s-k}+(\theta+1) p^{s-k-1}}) \geq n \geq 2,$$
it follows that
$$\operatorname{wt_{sp}}(  \left(x^n-\alpha_{0}\right)^{p^{s}-p^{s-k}+(\theta+1) p^{s-k-1}})=2 (\theta+2) p^k.$$
Thus $$\operatorname{d_{sp}}\left(\mathcal{C}_{p^{s}-p^{s-k}+(\theta+1) p^{s-k-1}}\right) \leq 2(\theta+2) p^{k}.$$
\end{proof}

The lower bounds of symbol-pair distances of $\mathcal{C}_{p^{s}-p^{s-k}+\theta p^{s-k-1}+1}$ is more complicated and we consider it in four subcases:
\begin{enumerate}[(i)]
\item $\theta=0$, $k=0$;
\item $\theta=0$, $1 \leq k \leq s-2$;
\item  $1 \leq \theta \leq p-2$, $0 \leq k \leq s-2$;
\item  $0 \leq \theta \leq p-2$, $k=s-1$.
\end{enumerate}

We start with the first case of $k=0$ and $\theta=0$.

\begin{lemma}\label{prop3.1} The pair distance $\operatorname{d_{sp}}(\mathcal{C}_1)$ of $\mathcal{C}_1=\langle x^n-\lambda_0\rangle\subseteq \mathbb{F}_{p^m}[x]/\langle x^{np^s}-\lambda\rangle$ is greater than or equal to $4$.
\end{lemma}
\begin{proof}
Verify that a codeword with symbol-pair weight two must be of form $u x^j$, which is invertible in $\mathcal{F}$. Hence there is no codeword in $\mathcal{C}_1$ with symbol-pair weight two. Note that a codeword with symbol-pair weight three has the form $u_0 x^j + u_1 x^{j+1}$, where $0 \leq j \leq n p^{s} -1$. It follows that $x^n-\alpha_0$ divides $u_0 x^j + u_1 x^{j+1}=(u_0 +u_1 x) x^j$, hence $x^n - \alpha_0 $ divides $ u_0+u_1 x$, which is impossible since the degree of $x^n - \alpha_0$ is greater than that of $u_0+u_1 x$. Hence there is no codeword in $\mathcal{C}_1$ with symbol-pair weight three. Therefore, we obtain $\operatorname{d_{sp}}(\mathcal{C}_1) \geq 4$.
\end{proof}

The following lemma shows the lower bound of the minimum pair-distance of $\mathcal{C}_{p^s- p^{s-k} +1}$ in the case that $\theta=0$ and $1 \leq k \leq s-2$.
\begin{lemma}\label{lem3.2}
Let $n,k$ be integers such that $n \geq 2$ and $1 \leq k \leq s-2$. Then $\operatorname{d_{sp}}(\mathcal{C}_{p^s- p^{s-k} +1}) \geq 4p^k$.
\end{lemma}
\begin{proof}

Let $c(x)$ be any nonzero codeword in $\mathcal{C}_{p^s - p^{s-k}+1}$. Then there is a nonzero element $f(x)$ in $\mathcal{F}$ such that $c(x)=\left(x^n-\alpha_{0}\right)^{p^{s}-p^{s-k}+1} f(x)$ with $\operatorname{deg}(f) < n p^s - n (p^{s}-p^{s-k}+1) = n( p^{s-k}-1)$. Let $g(x)=\left(x^n-\alpha_{0}\right) f(x)$. Then $\operatorname{deg}(g) < n p^{s-k}$, $\operatorname{wt_H}(g(x)) \geq 2$, and
$$\begin{aligned}
c(x) &=\left(x^n-\alpha_{0}\right)^{p^{s}-p^{s-k}} g(x) \\
&=\left[\sum_{j=0}^{p^{k}-1} \binom{p^{k} -1}{j} \left(-\alpha_{0}\right)^{p^{s-k}\left(p^{k}-j-1\right)} x^{n p^{s-k} j}\right] g(x).
\end{aligned}$$
We discuss the symbol-pair weight of $c(x)$ in the following three cases.

Case 1: If $\operatorname{deg}(g) \leq n p^{s-k}-2$, then
$$\operatorname{cw}( \left(x^n-\alpha_{0}\right)^{p^{s}-p^{s-k}} )=n p^{s-k} \geq  \operatorname{deg}(g) +2$$
and
$$\operatorname{deg}( \left(x^n-\alpha_{0}\right)^{p^{s}-p^{s-k}} ) + \operatorname{deg}(g) \leq n p^s -2.$$
By equation (\ref{equ2.2}), we have
$$\begin{array}{rl}
\operatorname{wt_{sp}}\left( c(x) \right) & =\operatorname{wt_H}( \left(x^n-\alpha_{0}\right)^{p^{s}-p^{s-k}}) \cdot \operatorname{wt_{sp}}\left(g(x)\right)\\
&= p^k \operatorname{wt_{sp}}\left(g(x)\right).\\
\end{array}$$
According to Lemma \ref{prop3.1}, $\operatorname{wt_{sp}}(g(x)) \geq \operatorname{d_{sp}}(\mathcal{C}_1) \geq 4$, which deduces that $\operatorname{wt_{sp}}\left( c(x) \right) \geq 4 p^k$.

Case 2: If $\operatorname{deg}(g)=n p^{s-k}-1$ and $g(0)=0$, then there is an integer $l > 0$ such that $g(x)=x^l g'(x)$, where $\operatorname{deg}(g') \leq n p^{s-k}-2$. Clearly,
$$\begin{array}{rl}
\operatorname{wt_{sp}}(c(x))&=\operatorname{wt_{sp}}( \left(x^n-\alpha_{0}\right)^{p^{s}-p^{s-k}} g(x) )\\
&=\operatorname{wt_{sp}}( \left(x^n-\alpha_{0}\right)^{p^{s}-p^{s-k}} x^l g'(x) )\\
&=\operatorname{wt_{sp}}( \left(x^n-\alpha_{0}\right)^{p^{s}-p^{s-k}} g'(x) ).
\end{array}$$
 Similar to the proof in Case 1, we have $\operatorname{wt_{sp}}\left( c(x) \right) \geq 4 p^k$.

Case 3: If $\operatorname{deg}(g)\!=\!n p^{s-k}-1$ and $g(0) \neq 0$, then $g(x)\!=\!\left(x^n-\alpha_{0}\right) f(x)$ is an element in $\langle x^n-\alpha_{0} \rangle$ of the ring ${\mathbb{F}_{p^{m}}[x]}/{\langle x^{np^{s-k}}-{\alpha_0}^{np^{s-k}}\rangle}$, $i.e.$, a codeword of an $\alpha_{0}^{np^{s-k}}$- constacyclic code of length $np^{s-k}$ over $\mathbb{F}_{p^{m}}$. According to Lemma \ref{prop3.1}, $\operatorname{wt_{sp}}(g(x))\geq 4$, which implies that $g(x)$ cannot be the form $r_0 + r_{1} x^{n p^{s-k}-1}$, where $r_0 , r_1 \neq 0$. Hence $\operatorname{wt_H}(g(x)) \geq 3$. When $\operatorname{wt_H}(g(x)) \geq 4$, we have
$$\begin{array}{rl}
\operatorname{wt_{sp}}(c(x)) & \geq \operatorname{wt_{H}}(c(x)) \\
& = \operatorname{wt_{H}}( \left(x^n-\alpha_{0}\right)^{p^{s}-p^{s-k}} ) \cdot \operatorname{wt_{H}} (g(x))\\
& \geq 4 p^k.
\end{array}$$
When $\operatorname{wt_H}(g(x))=3$, we assume that
 \begin{center}
 $g(x)=r_0 + r_1 x^l + r_{2} x^{n p^{s-k}-1}$,
\end{center}
where $0 < l < n p^{s-k}-1$ and $r_0 , r_1, r_2 \neq 0$. Let $S_H$ be a set of the exponents of nonzero terms of $c(x)$. Then the minimal partition of the set $S_H$ to subsets of consecutive indices may be the following three cases:\\
if $l=1$,
\[S_H=\cup_{1 \leq j \leq p^k-1}\{ np^{s-k}j-1, np^{s-k}j, np^{s-k}j+1\}\cup \{ 0,1, np^s-1\};\]
if $l=np^{s-k}-2$,
\[S_H=\cup_{1 \leq j \leq p^k-1}\{ np^{s-k}j-2, np^{s-k}j-1, np^{s-k}j \} \cup  \{ 0,np^s-1,np^s-2\};\]
if $1 < l < n p^{s-k}-2$,
\[S_H=\cup_{1 \leq j \leq p^k-1}(\{ np^{s-k}j-1, np^{s-k}j\} \cup \{np^{s-k}j+l\})\cup \{ 0,np^s-1\} \cup \{l\}.\]

According to the above three cases, we have $\operatorname{d_{sp}}(\mathcal{C}_{p^s-p^{s-k}+1}) \geq 4 p^k$. This completes the proof.
\end{proof}

The following lemma is considering the case of $0 \leq k \leq s-2$ and $1 \leq \theta \leq p-2$.
\begin{lemma}\label{lem3.3}
Let $n,k,\theta$ be integers such that $n \geq 2$, $0 \leq k \leq s-2$, and $1 \leq \theta \leq p-2$. Then $\operatorname{d_{sp}}(\mathcal{C}_{p^s- p^{s-k}+ \theta p^{s-k-1} +1}) \geq 2(\theta +2)p^k$.
\end{lemma}
\begin{proof}

Let $c(x)$ be any nonzero codeword in $\mathcal{C}_{p^s- p^{s-k}+ \theta p^{s-k-1} +1}$. Then there is a nonzero element $f(x)$ in $\mathcal{F}$ such that $c(x)\!\!=\!\!\left(x^n\!-\alpha_{0}\right)^{p^s- p^{s-k}+ \theta p^{s-k-1} +1}\!f(x)$ with $\operatorname{deg}(f) < n[ (p-\theta)p^{s-k-1}-1]$. Let $g(x)=\left(x^n-\alpha_{0}\right) f(x)$. Then $\operatorname{deg}(g) < n (p-\theta)p^{s-k-1}$, $\operatorname{wt_H}(g(x)) \geq 2$, and
$$\begin{array}{rl}
c(x)=&\left(x^{n}-\alpha_{0}\right)^{p^s- p^{s-k}+ \theta p^{s-k-1}} g(x) \\
= & (x^{n p^{s-k-1}}-\alpha_{0}^{p^{s-k-1}})^{p^{k+1}- p+ \theta} g(x).
\end{array}$$
Suppose that $\mathcal{T}=\{i_1, \cdots, i_{\eta}\}$ is a set of the exponents of nonzero terms of $g(x)$. For an integer $i$, let $\mathcal{S}_{i}$ be a set of integers congruent to $i$ modulo $n p^{s-k-1}$, $i.e.$, $\mathcal{S}_{i}=\{j \mid j \equiv i \; ( \bmod\; n p^{s-k-1})  \}$. We consider two cases that $\mathcal{T} \subset \mathcal{S}_{i_1}$ and $\mathcal{T} \not\subset \mathcal{S}_{i_1}$.

Case 1: When $\mathcal{T} \subset \mathcal{S}_{i_1}$. We assume that $g(x)=\sum_{t=1}^{\eta} r_{t} x^{i_{1}+n p^{s-k-1} u_{t}}$, where $0=u_{1}<\ldots<u_{\eta}$. Thus
$$\begin{array}{rl}
c(x)=&(x^{n p^{s-k-1}}-\alpha_{0}^{p^{s-k-1}})^{p^{k+1}- p+ \theta} g(x) \\
=& \left[ (x^{n p^{s-k-1}}-\alpha_{0}^{p^{s-k-1}})^{p^{k+1}- p+ \theta} \sum_{t=1}^{\eta} r_{t} x^{n p^{s-k-1} u_{t}} \right] x^{i_1}.
\end{array}$$
It follows that $\operatorname{cw}(c(x)) \geq n p^{s-k-1} \geq np \geq 4$.  Hence
$$\operatorname{wt_{sp}}(c(x))=2 \cdot \operatorname{wt_H}(c(x)) \geq 2 \cdot \operatorname{d_H}(\mathcal{C}_{p^s- p^{s-k}+ \theta p^{s-k-1} +1}) \geq 2(\theta+2) p^k .$$

Case 2: When $\mathcal{T} \not\subset \mathcal{S}_{i_1}$.  We only show that when $\mathcal{T} \subset \mathcal{S}_{i_1}\cup\mathcal{S}_{i_2}$ with $ i_1  \not\equiv i_2 \; (\bmod\; n p^{s-k-1})$, the rest is similar. Let $g(x)=g_1(x) + g_2(x)$,  $g_1(x)=\sum_{t=1}^{\eta_1} r_t^{(1)} x^{i_1 + n p^{s-k-1} u_t}$ and $g_2(x)=\sum_{t=1}^{\eta_2} r_t^{(2)} x^{i_2 + n p^{s-k-1} v_t}$, where $0=u_1 < \ldots < u_{\eta_1}$ and $0=v_1 < \ldots < v_{\eta_2}$. Then
$$\begin{array}{rl}
c(x)=&(x^{n p^{s-k-1}}-\alpha_{0}^{p^{s-k-1}})^{p^{k+1}- p+ \theta} \left( g_1(x) + g_2(x) \right) \\
=& \left[ (x^{n p^{s-k-1}}-\alpha_{0}^{p^{s-k-1}})^{p^{k+1}- p+ \theta}\sum_{t=1}^{\eta_1} r_t^{(1)} x^{ n p^{s-k-1} u_t} \right] x^{i_1} \\
& + \left[ (x^{n p^{s-k-1}}-\alpha_{0}^{p^{s-k-1}})^{p^{k+1}- p+ \theta}\sum_{t=1}^{\eta_2} r_t^{(2)} x^{n p^{s-k-1} v_t} \right] x^{i_2} .\\
\end{array}$$
Let $S_H$ be a set of the exponents of nonzero terms of $c(x)$. Then
$$\begin{array}{rl}
S_H=& \{ i_1+np^{s-k-1} w_j \mid w_j^{(1)} \in \mathbb{N}, 1 \leq j \leq l_1 \}\\
& \cup \{ i_2+np^{s-k-1} w_j \mid w_j^{(2)} \in \mathbb{N}, 1 \leq j \leq l_2 \},
\end{array}$$
where $l_t=\operatorname{wt_H}(\left(x^{n}-\alpha_{0}\right)^{p^s- p^{s-k}+ \theta p^{s-k-1}} g_t(x))$ for $t=1,2$. By Theorem \ref{th_H_F}, $\operatorname{d_H}(\mathcal{C}_{p^s- p^{s-k}+ \theta p^{s-k-1}}) \geq (\theta+1) p^k$, and hence $l_1,l_2 \geq (\theta+1) p^k$. Since $np^{s-k-1} \geq 4$, $S_H$ is at least partitioned into $(\theta+1) p^k$ subsets of consecutive indices. By Theorem \ref{th_H_S},
$$\operatorname{wt_{sp}}(c(x)) \geq 3 (\theta+1) p^k \geq 2(\theta+2) p^k.$$
Therefore, we have proved that $\operatorname{wt_{sp}}(c(x)) \geq 2(\theta+2) p^k$ holds in all cases, that is, $\operatorname{d_{sp}}(\mathcal{C}_{p^s- p^{s-k}+ \theta p^{s-k-1} +1}) \geq 2(\theta +2)p^k$.

\end{proof}

The following lemma is about the case of $k=s-1$ and $0 \leq \theta \leq p-2$.
\begin{lemma}\label{lem3.4}
Let $n,\theta$ be integers such that $n \geq 2$ and $1 \leq \theta \leq p-1$. Then $\operatorname{d_{sp}}(\mathcal{C}_{p^s-p+\theta}) \geq 2(\theta + 1)p^{s-1}$.
\end{lemma}
\begin{proof}
Let $c(x)$ be any nonzero codeword in $\mathcal{C}_{p^s- p+ \theta}$. Then there is a nonzero element $f(x)$ in $\mathcal{F}$ such that $c(x)=\left(x^n-\alpha_{0}\right)^{p^s- p+ \theta} f(x)$ with $\operatorname{deg}(f) < n(p-\theta)$. Suppose that $\mathcal{T}=\{i_1, \cdots, i_{\eta}\}$ is a set of the exponents of nonzero terms of $f(x)$. For an integer $i$, let $\mathcal{S}_{i}$ be a set of integers congruent to $i$ modulo $n$, $i.e.$, $\mathcal{S}_{i}=\{j \mid j \equiv i \; ( \bmod\; n )  \}$. We consider the set $\mathcal{T}$ in two cases.

Case 1: When $\mathcal{T} \subset \mathcal{S}_{i_1}$. We may assume that $f(x)=\sum_{t=1}^{\eta} r_{t} x^{i_{1}+n u_{t}}$, where $0=u_{1}<\ldots<u_{\eta}$. Then
$$c(x)=\left[\left(x^n-\alpha_{0}\right)^{p^s- p+ \theta} \sum_{t=1}^{\eta} r_{t} x^{n u_{t}}\right] x^{i_1}.$$
It follows that $\operatorname{cw}(c(x)) \geq n \geq 2$, and hence,
$$\operatorname{wt_{sp}}(c(x))=2 \cdot \operatorname{wt_H}(c(x)) \geq 2 \cdot \operatorname{d_H}(\mathcal{C}_{p^s- p+ \theta}) \geq 2(\theta+1) p^{s-1}.$$

Case 2: When $\mathcal{T} \not\subset \mathcal{S}_{i_1}$. We may assume that $\mathcal{T} \subset \mathcal{S}_{i_1}\cup\mathcal{S}_{i_2}$, where $i_1 \not\equiv i_2 \; (\bmod\; n)$. Let $f(x)=f_1(x) + f_2(x)$,  $f_1(x)=\sum_{t=1}^{\eta_1} r_t^{(1)} x^{i_1 + n u_t}$, and $f_2(x)=\sum_{t=1}^{\eta_2} r_t^{(2)} x^{i_2 + n v_t}$, where $0=u_1 < \ldots < u_{\eta_1}$ and $0=v_1 < \ldots < v_{\eta_2} $. Then
$$\begin{array}{rl}
c(x)= & \left[\left(x^n-\alpha_{0}\right)^{p^s- p+ \theta} \sum_{t=1}^{\eta_1} r_{t}^{(1)} x^{n u_{t}}\right] x^{i_1}\\
&+ \left[\left(x^n-\alpha_{0}\right)^{p^s- p+ \theta} \sum_{t=1}^{\eta_2} r_{t}^{(2)} x^{n v_{t}}\right] x^{i_2}.
\end{array}$$
Since $ i_1  \not\equiv i_2 \; (\bmod\; n)$,
$$\begin{array}{rl}
\operatorname{wt_H}(c(x))= & \operatorname{wt_H}\left([\left(x^n-\alpha_{0}\right)^{p^s- p+ \theta} \sum_{t=1}^{\eta_1} r_{t}^{(1)} x^{n u_{t}}] x^{i_1}\right)\\
&+ \operatorname{wt_H}\left([\left(x^n-\alpha_{0}\right)^{p^s- p+ \theta} \sum_{t=1}^{\eta_2} r_{t}^{(2)} x^{n v_{t}}] x^{i_2}\right)\\
\geq & 2 \cdot \operatorname{d_H}(\mathcal{C}_{p^s- p+ \theta})\\
= & 2(\theta+1) p^{s-1},
\end{array}$$
which implies that $\operatorname{wt_{sp}}(c(x)) \geq \operatorname{wt_H}(c(x)) \geq 2(\theta+1)p^{s-1}$. Combining the two cases discussed above, it follows that $\operatorname{d_{sp}}(\mathcal{C}_{p^s-p+\theta}) \geq 2(\theta+1)p^{s-1}$.
\end{proof}

Combining the upper bound given in Lemma \ref{lem3.1} and the lower bounds given in Lemma \ref{prop3.1}, \ref{lem3.2}, \ref{lem3.3}, and \ref{lem3.4}, the symbol-pair distances of $\alpha$-constacyclic codes of length $np^s$ over $\mathbb{F}_{p^m}$ can be completely determined. In order to maintain the integrity of the theorem, we present the symbol-pair distances for both the cases that $n=1$ and $n\geq 2$.
\begin{theorem}\label{th3.1}
Let $\alpha_{0}$ be a nonzero element in $\mathbb{F}_{p^m}$ and $\alpha=\alpha_{0}^{p^s}$. Given an $\alpha$-constacyclic code of length $np^s$ over $\mathbb{F}_{p^m}$. If it has the form as $\mathcal{C}_i=\langle\left(x^n-\alpha_{0}\right)^{i}\rangle \subseteq \mathcal{F}$, for $i \in\left\{0,1, \ldots, p^{s}\right\}$, then the symbol-pair distance $\operatorname{d_{sp}}(\mathcal{C}_i)$ is completely determined by:
\begin{enumerate}[(i)]
\item
(Trivial cases) $\operatorname{d_{sp}}(\mathcal{C}_0)=2$ and $\operatorname{d_{sp}}(\mathcal{C}_{p^s})=0$.
\item
When $n=1$,
$$\operatorname{d_{sp}}(\mathcal{C}_{i})=\left\{\begin{array}{ll}
3 p^k, & \text{ if } i=p^s-p^{s-k}+1 \text{ and }  0 \leq k \leq s-2;\\
4 p^k, & \text{ if } p^{s}-p^{s-k}+2 \leq i \leq p^{s}-p^{s-k}+p^{s-k-1}\\
& \text{ and }  0 \leq k \leq s-2;\\
2(\theta +2 )p^{k}, & \text{ if } p^s - p^{s-k}+ \theta p^{s-k-1}+1 \leq i \leq \\
&\quad p^s - p^{s-k}+ (\theta+1) p^{s-k-1},\\
&\quad 0 \leq k \leq s-2 \text{ and }  1 \leq \theta \leq p-2;\\
(\theta+2)p^{s-1}, & \text{ if } i=p^s-p+\theta \text{ and } 1 \leq \theta \leq p-2;\\
p^s, & \text{ if } i=p^s-1.
\end{array}\right.$$
\item
When $n \geq 2$, \[\operatorname{d_{sp}}(\mathcal{C}_i)=2(\theta +2)p^{k},\]
  where $p^s - p^{s-k}+ \theta p^{s-k-1}+1 \leq i \leq p^s - p^{s-k}+ (\theta+1) p^{s-k-1}$, $0 \leq k \leq s-1$ and $0 \leq \theta \leq p-2$.
\end{enumerate}

\end{theorem}

\subsection{MDS codes}

In Table \ref{tab1}, we summarize the MDS symbol-pair codes given in the previous literature. In this subsection , we use the result of symbol-pair distances obtained in the former subsection to prove that when $x^n-\alpha_0$ is irreducible over $\mathbb{F}_{p^m}$, there are no other MDS symbol-pair codes except for these in Table~\ref{tab1}.

\begin{theorem}\label{th3.2}
Let $\alpha_{0}$ be a nonzero element of $\mathbb{F}_{p^m}$ and $\alpha=\alpha_{0}^{p^s}$. When $x^n-\alpha_0$ is irreducible over $\mathbb{F}_{p^m}$, there are no other nontrivial MDS symbol-pair $\alpha$-constacyclic codes of length $n p^s$ over $\mathbb{F}_{p^m}$ except the MDS codes shown in Table \ref{tab1}.
\end{theorem}
\begin{proof}
When $x^n-\alpha_0$ is irreducible over $\mathbb{F}_{p^m}$, the $\alpha$-constacyclic codes of length $n p^s$ over $\mathbb{F}_{p^m}$ are $\mathcal{C}_i=\langle\left(x^{n}-\alpha_{0}\right)^{i}\rangle$, where $0 \leq i \leq p^s$. Note that $| \mathcal{C}_i |=| \langle (x^n-\alpha_0)^i \rangle |= p^{m(np^s-ni)}$ and the Singleton Bounds for symbol-pair constacyclic codes force $| \mathcal{C}_i | \leq p^{m(np^s-\operatorname{d_{sp}}(\mathcal{C}_i)+2)}$, $i.e.$, $ni \geq \operatorname{d_{sp}}(\mathcal{C}_i)-2$ for $i \in \{0,1,\ldots,p^s-1\}$. Therefore, $\mathcal{C}_i$ is an MDS symbol-pair code if and only if $ni-\operatorname{d_{sp}}(\mathcal{C}_i)+2=0$. If $n$ is equal to 1 or 2, all the MDS codes have been constructed by Dinh et al.\cite{Dinh2018,Dinh2019a} and listed in Table \ref{tab1}. If $n \geq 3$, let $i=p^s - p^{s-k}+ \theta p^{s-k-1}+\gamma$ with $0 \leq k \leq s-1$, $0 \leq \theta \leq p-2$, and $1 \leq \gamma \leq p^{s-k-1}$. By Theorem \ref{th3.1}, we have $\operatorname{d_{sp}}(\mathcal{C}_i)=2(\theta +2 )p^{k}$, hence
$$\begin{aligned}
&ni-\operatorname{d_{sp}}(\mathcal{C}_{i})+2\\
=&n(p^s-p^{s-k}+\theta p^{s-k-1}+\gamma)-2(\theta+2)p^k+2\\
=&[np^{s-k}-2(\theta+2)](p^k-1)+(np^{s-k-1}-2)\theta + n\gamma-2\\
\geq& n-2 >0.
\end{aligned}$$
Therefore, there is no other MDS symbol-pair $\alpha$-constacyclic code.
\end{proof}

\section{MDS Symbol-pair codes over \texorpdfstring{$\mathbb{F}_{p^m} + u \mathbb{F}_{p^m}$}{TEXT}}\label{sec4}

\subsection{The pair distances of constacyclic codes over $\mathbb{F}_{p^m} + u \mathbb{F}_{p^m}$}

Let $\alpha_{0},\beta\in\mathbb{F}_{p^m}$ and $\alpha_{0}\neq 0$. Denote $\alpha$ as $\alpha=\alpha_{0}^{p^s}$. In this section, we characterize the relationship between the symbol-pair distances of $\alpha$-constacyclic codes of length $n p^s$ over $\mathbb{F}_{p^m}$ and the symbol-pair distances of $(\alpha+u\beta)$-constacyclic codes of length $n p^s$ over $\mathbb{F}_{p^m} + u \mathbb{F}_{p^m}$, where $n$ is a positive integer coprime to $p$ and $x^n-\alpha_{0}$ is irreducible over $\mathbb{F}_{p^{m}}$.
We analyze the symbol-pair distances in the cases that $\beta\neq 0$ and $\beta = 0$. When $\beta\neq 0$, $\mathcal{R}=(\mathbb{F}_{p^m}+u\mathbb{F}_{p^m})[x] /\left\langle x^{n p^{s}}-\alpha-u \beta\right\rangle$ is a chain ring and all the ideals of $\mathcal{R}$ are $\mathcal{D}_{i}=\langle (x^n-\alpha_0)^i\rangle$, where $0\leq i\leq 2p^s$.
\begin{theorem}\label{th4.1}
Let $\alpha_{0}$ be a nonzero element in $\mathbb{F}_{p^m}$ satisfying that $x^n-\alpha_{0}$ is irreducible over $\mathbb{F}_{p^m}$. Denote $\alpha=\alpha_{0}^{p^s}$. Let $\beta$ be a nonzero element in $\mathbb{F}_{p^m}$. The symbol-pair distance of $\mathcal{D}_{i}=\langle (x^n-\alpha_0)^i\rangle$ is
$$\operatorname{d_{sp}}(\mathcal{D}_i)=\left\{\begin{array}{ll}
2, & \text{if } 0\leq i\leq p^s;\\
\operatorname{d_{sp}}(\langle (x^n-\alpha_{0})^{i-p^{s}}\rangle_{F}), & \text{if } p^s+1 \leq i\leq 2p^s. \\
\end{array}\right.$$
\end{theorem}
\begin{proof}
When $0 \leq i \leq p^{s}$, we have $u(x)=\left(x^{n}-\alpha_{0}\right)^{p^s} \in\mathcal{D}_i$ and the symbol-pair weight of $u(x)$ is $2$. Combining with $\operatorname{d_{sp}}(D_{i})\geq 2$, we have the symbol-pair distance of $\mathcal{D}_i$ is $2$.

When $p^{s}+1 \leq i \leq 2 p^{s}$, we have $$\langle(x^{n}-\alpha_{0})^{i}\rangle=\langle u(x^{n}-\alpha_{0})^{i-p^{s}}\rangle,$$
which means that the codewords in the code $\langle(x^{n}-\alpha_{0})^{i}\rangle$ over $\mathbb{F}_{p^m} + u \mathbb{F}_{p^m}$ are precisely the codewords in the code $\langle(x^{n}-\alpha_{0})^{i-p^{s}}\rangle$ over $\mathbb{F}_{p^m}$ multiplied with $u$. Therefore, the symbol-pair distance of $\mathcal{D}_{i}$ is equal to that of $\langle (x^n-\alpha_{0})^{i-p^{s}}\rangle_{F}$.
\end{proof}

The symbol-pair distance of $\mathcal{D}_{i}$ is more complicated when $\beta=0$, and we analysis it in three cases according to the three types of $\alpha$-constacyclic codes shown in Theorem \ref{th_R_0}.

\begin{theorem}\label{th4.2}
Let $\mathcal{D}$ be an $\alpha$-constacyclic code of length $n p^{s}$ over $\mathbb{F}_{p^m} + u \mathbb{F}_{p^m}$ with type I in Theorem \ref{th_R_0}, $i.e.$, $\mathcal{D}=\langle \left(x^n-\alpha_{0}\right)^{k}\rangle$ for $0 \leq k \leq p^{s}$. Then $\operatorname{d_{sp}}(\mathcal{D})=\operatorname{d_{sp}}(\langle\left(x^n-\alpha_{0}\right)^{k}\rangle_{F})$.
\end{theorem}
\begin{proof}
Notice that $\mathcal{D} \supseteq \langle u \left(x^n-\alpha_{0}\right)^{k}\rangle$, and hence
\begin{equation}\label{th4.2a}
\operatorname{d_{sp}}(\mathcal{D}) \leq \operatorname{d_{sp}}(\langle u \left(x^n-\alpha_{0}\right)^{k}\rangle) = \operatorname{d_{sp}}(\langle\left(x^n-\alpha_{0}\right)^{k}\rangle_{F}).
\end{equation}

Next, for any nonzero codeword $c(x)$ in $\mathcal{D}$, there are $f_0(x),f_u(x)$ in $\mathbb{F}_{p^m}[x]$ such that
$$\begin{array}{rl}
c(x)& =[f_0(x) + uf_u(x)] \left(x^n-\alpha_{0}\right)^{k}\\
&=f_0(x) \left(x^n-\alpha_{0}\right)^{k} + u f_u(x) \left(x^n-\alpha_{0}\right)^{k} .
\end{array}$$
It follows that
\begin{align}
\operatorname{wt_{sp}}(c(x))& \geq \max\{ \operatorname{wt_{sp}}(f_0(x) \left(x^n-\alpha_{0}\right)^{k}), \operatorname{wt_{sp}}(f_u(x) \left(x^n-\alpha_{0}\right)^{k}) \} \notag\\
& \geq \operatorname{d_{sp}}(\langle\left(x^n-\alpha_{0}\right)^{k}\rangle_{F}). \label{th4.2b}
\end{align}
Combining \eqref{th4.2a} and \eqref{th4.2b}, we have
$$\operatorname{d_{sp}}(\mathcal{D})=\operatorname{d_{sp}}(\langle\left(x^n-\alpha_{0}\right)^{k}\rangle_{F}).$$

\end{proof}

The following theorem shows the symbol-pair distances of the constacyclic codes corresponding to the second type.

\begin{theorem}\label{th4.3}
Let $\mathcal{D}$ be an $\alpha$-constacyclic code of length $n p^{s}$ over $\mathbb{F}_{p^m} + u \mathbb{F}_{p^m}$ with type II in Theorem \ref{th_R_0}, $i.e.$, $\mathcal{D}=\langle(x^{n}-\alpha_{0})^{j} b(x)+u(x^{n}-\alpha_{0})^{k}\rangle$, where $0 \leq k \leq p^s-1$, $\left\lceil\frac{p^{s}+k}{2}\right\rceil \leq j \leq p^s-1$ and either $b(x)$ is 0 or $b(x)$ is a unit in $\mathcal{F}$. Then
$$\operatorname{d_{sp}}(\mathcal{D})=\left\{ \begin{array}{ll}
\operatorname{d_{sp}}(\langle(x^n-\alpha_{0})^{k}\rangle_{F}),& \text{if } b(x)=0; \\
\operatorname{d_{sp}}(\langle(x^n-\alpha_{0})^{p^s-j+k}\rangle_{F}),& \text{if } b(x) \text{ is a unit in } \mathcal{F}. \\
\end{array} \right.$$
\end{theorem}
\begin{proof}
If $b(x)=0$, then $\mathcal{D}=\langle u\left(x^{n}-\alpha_{0}\right)^{k}\rangle$. Hence $$\operatorname{d_{sp}}(\mathcal{D})=\operatorname{d_{sp}}(\langle u\left(x^n-\alpha_{0}\right)^{k}\rangle)=\operatorname{d_{sp}}(\langle\left(x^n-\alpha_{0}\right)^{k}\rangle_{F}).$$

Assume that $b(x)$ is a unit in $\mathcal{F}$. Since
$$\left(x^{n}-\alpha_{0}\right)^{p^s-j}\left[ \left(x^{n}-\alpha_{0}\right)^{j} b(x)+u\left(x^{n}-\alpha_{0}\right)^{k} \right]=u \left(x^{n}-\alpha_{0}\right)^{p^s-j+k},$$
it follows that
$$\langle u\left(x^n-\alpha_{0}\right)^{p^s-j+k}\rangle \subseteq \mathcal{D},$$
hence
\begin{align}
\operatorname{d_{sp}}(\mathcal{D}) \leq  & \operatorname{d_{sp}}(\langle u\left(x^n-\alpha_{0}\right)^{p^s-j+k}\rangle) \notag\\
= & \operatorname{d_{sp}}(\langle\left(x^n-\alpha_{0}\right)^{p^s-j+k}\rangle_{F}).\label{th4.3 1}
\end{align}
For any nonzero codeword $c(x)$ in $\mathcal{D}$, there are $f_0(x)$, $f_u(x)$ in $\mathbb{F}_{p^m}[x]$ such that
$$\begin{aligned}
c(x)& =[f_0(x) + uf_u(x)] \left[ \left(x^{n}-\alpha_{0}\right)^{j} b(x)+u\left(x^{n}-\alpha_{0}\right)^{k} \right]\\
& =f_0(x)\left(x^{n}-\alpha_{0}\right)^{j} b(x) + u \left[ f_0(x) \left(x^{n}-\alpha_{0}\right)^{k} + f_u(x) \left(x^{n}-\alpha_{0}\right)^{j} b(x) \right].
\end{aligned}$$
It follows that
$$\begin{aligned}
\operatorname{wt_{sp}}(c(x)) \geq &\max\{   \operatorname{wt_{sp}}(f_0(x)\left(x^{n}-\alpha_{0}\right)^{j} b(x)),\\
& \operatorname{wt_{sp}}(f_0(x) \left(x^{n}-\alpha_{0}\right)^{k} + f_u(x) \left(x^{n}-\alpha_{0}\right)^{j} b(x)) \}.
\end{aligned}$$
If $\left(x^{n}-\alpha_{0}\right)^{p^s-j} \mid f_0(x)$, let $f_0(x)=\left(x^{n}-\alpha_{0}\right)^{p^s-j}f_0^{'}(x)$. Then
\begin{align}
\operatorname{wt_{sp}}(c(x)) & \geq \operatorname{wt_{sp}}(r(x)\left(x^{n}-\alpha_{0}\right)^{p^s-j+k}) \notag\\
& \geq \operatorname{d_{sp}}(\langle(x^n-\alpha_{0})^{p^s-j+k}\rangle_{F}), \label{th4.3 2}
\end{align}
where
 \begin{center}
  $r(x)$ $\!=\!$ $f_0^{'}(x)$ $+$ $f_u(x) \left(x^{n}-\alpha_{0}\right)^{2j-p^s-k} b(x)$.
\end{center}
If $\left(x^{n}-\alpha_{0}\right)^{p^s-j} \nmid f_0(x)$, then
\begin{center}
$\operatorname{wt_{sp}}(c(x)) \geq \operatorname{wt_{sp}}(f_0(x)\left(x^{n}-\alpha_{0}\right)^{j} b(x))$.
\end{center}
Since $j \geq p^s-j+k$,
\begin{equation}\label{th4.3 3}
  \operatorname{wt_{sp}}(c(x)) \geq \operatorname{d_{sp}}(\langle(x^n-\alpha_{0})^{p^s-j+k}\rangle_{F}).
\end{equation}
According to \eqref{th4.3 1}, \eqref{th4.3 2} and \eqref{th4.3 3}, we have
\begin{center}
$\operatorname{d_{sp}}(\mathcal{D}) = \operatorname{d_{sp}}(\langle\left(x^n-\alpha_{0}\right)^{p^s-j+k}\rangle_{F})$.
\end{center}

\end{proof}

\begin{remark}\label{rem4.1}
Our results of symbol-pair distances of constacyclic codes over $\mathbb{F}_{p^m}+u\mathbb{F}_{p^m}$ generalize the results of \cite{Dinh2018a}, which consider the constacyclic codes under the condition of $n=1$.
\end{remark}

\begin{remark}\label{rem4.2}
According to Theorem \ref{th4.3}, when $b(x)$ is a unit, the symbol-pair distance of the constacyclic code $\mathcal{D}=\langle (x-\alpha_{0})^jb(x)+u(x-\alpha_0)^k\rangle$ with $k < 2j-p^s$ is equal to $\operatorname{d_{sp}}(\langle(x-\alpha_{0})^{p^s-j+k}\rangle_{F})$, but not as \cite{Dinh2018a} claimed that equal to $\operatorname{d_{sp}}(\langle(x-\alpha_{0})^{j}\rangle_{F})$. We illustrate an example to show this fact in the following.
\end{remark}

\begin{example}
Consider a cyclic code $\mathcal{D}=\langle (x-1)^7+u(x-1) \rangle$ of length $9$ over the finite ring $\mathbb{F}_{3}+u\mathbb{F}_{3}$, where $u^2=0$. By Theorem 12 in \cite{Dinh2018a},
$$\operatorname{d_{sp}}(\mathcal{D})=\operatorname{d_{sp}}(\langle(x-1)^{7}\rangle_{F})=9.$$
However, there is a codeword
$$u(x-1)^3=(x-1)^2[(x-1)^7+u(x-1)] \in \mathcal{D},$$
and
\begin{center}
$\operatorname{wt_{sp}}(u(x-1)^3)=4$,
\end{center}
which means the symbol-pair distance of $\mathcal{D}$ cannot be $9$.
Actually, according to Theorem \ref{th4.3}
\[\operatorname{d_{sp}}(\mathcal{D})=\operatorname{d_{sp}}(\langle (x-1)^3\rangle_F)=4.\]
\end{example}

The following theorem shows the symbol-pair distances of the constacyclic codes corresponding to the type III in Theorem \ref{th_R_0}. The proof is similar to that of the former theorem, and we omit it here.

\begin{theorem}\label{th4.4}
Let $\mathcal{D}$ be an $\alpha$-constacyclic code of length $n p^{s}$ over $\mathbb{F}_{p^m} + u \mathbb{F}_{p^m}$ with type III in Theorem \ref{th_R_0}, $i.e.$, $\mathcal{D}\!\!=\!\!\langle\!(x^{n}\!\!-\!\alpha_{0})^{j}b(x)\!+\!u(x^{n}\!\!-\!\alpha_{0})^{k}\!\!, (x^{n}\!\!-\!\alpha_{0})^{k+t}\rangle$, where $0 \leq k \leq p^s-2$, $1 \leq t \leq p^s-k-1$, $k + \left\lceil\frac{t}{2}\right\rceil \leq j \leq  k+t$, and either $b(x)$ is 0 or $b(x)$ is a unit in $\mathcal{F}$. Then
$$\operatorname{d_{sp}}(\mathcal{D})=\left\{\begin{array}{ll}
\operatorname{d_{sp}}(\langle(x^n-\alpha_{0})^{k}\rangle_{F}), & \text{if}~b(x)=0;\\
\operatorname{d_{sp}}(\langle(x^n-\alpha_{0})^{2k+t-j}\rangle_{F}), &\text{if}~b(x) \text{ is a unit in } \mathcal{F}. \\
\end{array} \right.$$
\end{theorem}

\subsection{MDS codes}

In this subsection, we utilize the symbol-pair distances of $(\alpha+u\beta)$-constacyclic codes of length $np^s$ over $\mathbb{F}_{p^m} + u \mathbb{F}_{p^m}$ shown in the former subsection to obtain the MDS symbol-pair codes.

The following theorem shows that no nontrivial MDS symbol-pair $(\alpha+u\beta)$-constacyclic code exists when $\beta\neq 0$.

\begin{theorem}\label{th4.5}
Let $\alpha_{0},\beta$ be nonzero elements in $\mathbb{F}_{p^m}$. Denote $\alpha=\alpha_{0}^{p^s}$. Suppose that $x^n-\alpha_{0}$ is irreducible over $\mathbb{F}_{p^m}$. Let $\mathcal{D}_i=\langle (x^n-\alpha_0)^i \rangle \subseteq \mathcal{R}$ be an $(\alpha+u\beta)$-constacyclic code of length $n p^s$ over $\mathbb{F}_{p^m} + u \mathbb{F}_{p^m}$, where $0\leq i\leq 2p^s$. Then $\mathcal{D}_i$ is an MDS symbol-pair code if and only if $i=0$.
\end{theorem}
\begin{proof}
By Theorem \ref{th_R_1},
$$|\mathcal{D}_{i}|=|\langle(x^{n}-\alpha_{0})^{i}\rangle|=p^{mn( 2p^{s}- i)}.$$
The Singleton Bound shows
$$|\mathcal{D}_{i}| \leq |R|^{np^s-\operatorname{d_{sp}}(\mathcal{C}_i)+2},$$
which is equivalent to
\begin{equation}
 ni \geq 2\operatorname{d_{sp}}(\mathcal{D}_i)-4. \notag
\end{equation}
Therefore, $\mathcal{D}_{i}$ is an MDS symbol-pair code if and only if
\begin{equation}\label{th4.5a}
  ni=2\operatorname{d_{sp}}(\mathcal{D}_i)-4.
\end{equation}
If $0 \leq i \leq p^{s}$, we have $\operatorname{d_{sp}}\left(\mathcal{D}_{i}\right)=2$ by Theorem \ref{th4.1}. According to \eqref{th4.5a}, we obtain $i=0$. If $p^s+1 \leq i \leq 2p^s-1$, then $\operatorname{d_{sp}}(\mathcal{D}_i)=\operatorname{d_{sp}}(\langle(x^n-\alpha_{0})^{i-p^s}\rangle_{F})$. Applying the Singleton bound on the constacyclic code $\langle(x^n-\alpha_{0})^{i-p^s}\rangle_{F}$, we obtain
\begin{equation}\label{th4.5b}
n(i-p^s) \geq \operatorname{d_{sp}}(\langle(x^n-\alpha_{0})^{i-p^s}\rangle_{F})-2.
\end{equation}
Reformulate \eqref{th4.5b} we have
\begin{align*}
  ni \geq & \operatorname{d_{sp}}(\mathcal{D}_i) -2 + np^s \geq  2\operatorname{d_{sp}}(\mathcal{D}_i) -2\\
  >&2\operatorname{d_{sp}}(\mathcal{D}_i) -4,
\end{align*}
which implies no MDS symbol-pair constacyclic code when $i$ is in the range $ p^s+1 \leq i \leq 2p^s-1$.

\end{proof}

The following theorem is considering the case of $\beta=0$ and we obtain three new classes of MDS symbol-pair $\alpha$-constacyclic codes over $\mathbb{F}_{p^m} + u \mathbb{F}_{p^m}$.

\begin{theorem}\label{th4.6}
Let $\alpha_{0}$ be a nonzero element in $\mathbb{F}_{p^m}$ and $\alpha=\alpha_{0}^{p^s}$. There are three classes of  MDS symbol-pair $\alpha$-constacyclic codes of length $2 p^{s}$ over $\mathbb{F}_{p^m} + u \mathbb{F}_{p^m}$ as follows:
\begin{enumerate}[(i)]
\item
$\mathcal{D}=\langle(x^{2}-\alpha_{0})+ u b(x)\rangle$, where $b(x)$ is either zero or a unit in $\mathcal{F}$;
\item
$\mathcal{D}=\langle(x^{2}-\alpha_{0})^{p^s -1}+ u (x^{2}-\alpha_{0})^{p^s -2} b(x)\rangle$, where $b(x)$ is either zero or a unit in $\mathcal{F}$;
\item
$\mathcal{D}=\langle(x^{2}-\alpha_{0})^{j}+ u (x^{2}-\alpha_{0})^{k} b(x)\rangle$, where $s=1$, $1 \leq j \leq p-1$, $max\{0,2j-p\} \leq k < j$, and $b(x)$ is either zero or a unit in $\mathcal{F}$.
\end{enumerate}
\end{theorem}
\begin{proof}
(i)
By Lemma \ref{th_R_0}, the size of $\mathcal{D}$ is $p^{4m(p^s-1)}$. According to Theorem \ref{th4.2} and \ref{th4.4}, the symbol-pair distance of $\mathcal{D}$ is 4. It achieves the Singleton bound
\[|\mathcal{D}|=p^{4m(p^s-1)}=p^{2m(2p^s -4 +2)}\]
with equality. Therefore, $\mathcal{D}$ is an MDS code.

(ii)
By Lemma \ref{th_R_0}, the size of $\mathcal{D}$ is $p^{4m}$. According to Theorem \ref{th4.2} and \ref{th4.3}, the symbol-pair distance of $\mathcal{D}$ is $2p^s$. It achieves the Singleton bound
\[|\mathcal{D}|=p^{4m}=p^{2m(2p^s -2p^s +2)}\]
with equality. Therefore, $\mathcal{D}$ is an MDS code.

(iii)
By Lemma \ref{th_R_0}, the size of $\mathcal{D}$ is $p^{4m(p-j)}$. According to Theorem \ref{th4.2}, \ref{th4.3} and \ref{th4.4}, the symbol-pair distance of $\mathcal{D}$ is $2j+2$. It achieves the Singleton bound
\[|\mathcal{D}|=p^{4m(p-j)}=p^{2m(2p -2j-2 +2)}\]
with equality. Therefore, $\mathcal{D}$ is an MDS code.
\end{proof}

\begin{remark}\label{rem4.3}
In \cite{Dinh2019b}, Dinh \emph{et al.} gave two more classes of MDS symbol-pair codes with parameters $(2^s, 2^{m(2^{s-1}+4)}, 3 \cdot 2^{s-2})$ and $(3^s, 3^{m(2 \cdot 3^{s-1}+4)}, 2 \cdot 3^{s-1})$.

The first $\alpha$-constacyclic code of length $2^{s}$ over $\mathbb{F}_{2^m} + u \mathbb{F}_{2^m}$ is
$$\mathcal{D}=\langle\left(x^{n}-\alpha_{0}\right)^{2^s-3}+u\left(x^{n}-\alpha_{0}\right)^{2^{s-1}-4}\rangle,$$
where $s \geq 3$. By Remark \ref{rem4.2}, the symbol-pair distance of $\mathcal{C}$ is $\mathrm{d}_{\mathrm{sp}}(\langle(x-\alpha_{0})^{2^{s-1}-1}\rangle_{F})=4$, but not $\mathrm{d}_{\mathrm{sp}}(\langle(x-\alpha_{0})^{2^s-3}\rangle_{F})=3 \cdot 2^{s-2}$. Hence it is not MDS.

The second $\alpha$-constacyclic code of length $3^{s}$ over $\mathbb{F}_{3^m} + u \mathbb{F}_{3^m}$ is
$$\mathcal{D}=\langle\left(x^{n}-\alpha_{0}\right)^{3^s-5}+u\left(x^{n}-\alpha_{0}\right)^{3^{s-1}-4}\rangle,$$
where $s \geq 3$. By Remark \ref{rem4.2}, the symbol-pair distance of $\mathcal{C}$ is $\mathrm{d}_{\mathrm{sp}}(\langle(x-\alpha_{0})^{3^{s-1}+1}\rangle_{F})=6$, but not $\mathrm{d}_{\mathrm{sp}}(\langle(x-\alpha_{0})^{3^s-5}\rangle_{F})=2 \cdot 3^{s-1}$. Hence it is also not MDS.
\end{remark}

Combining Theorem \ref{th4.6} and previous work, the MDS symbol-pair $\alpha$-constacyclic codes of length $n p^{s}$ over $\mathbb{F}_{p^m} + u \mathbb{F}_{p^m}$ are listed in Table \ref{tab2}. When the polynomial $x^{n}-\alpha_0$ is irreducible over $\mathbb{F}_{p^m}$, we draw the following conclusion. The proof is similar to Theorem \ref{th3.2} and  omitted here.

\begin{theorem}\label{th4.7}
Let $\alpha_{0}$ be a nonzero element of $\mathbb{F}_{p^m}$ and $\alpha=\alpha_{0}^{p^s}$. Let $\beta$ be an element of $\mathbb{F}_{p^m}$. When $x^n-\alpha_0$ is irreducible over $\mathbb{F}_{p^m}$, there are no other nontrivial MDS symbol-pair $\alpha+u\beta$-constacyclic codes of length $n p^s$ over $\mathbb{F}_{p^m}+u\mathbb{F}_{p^m}$ except the MDS codes shown in Table \ref{tab2}.
\end{theorem}

\section{Conclusion}

Let $\mathbb{F}_{p^m}$ be a finite field and $\mathbb{F}_{p^m}+u\mathbb{F}_{p^m}$ be a finite ring with $u^2=0$. We determine the symbol-pair distances of $\alpha$-constacyclic codes of length $np^s$ over $\mathbb{F}_{p^m}$ and $(\alpha+u\beta)$-constacyclic codes of length $np^s$ over $\mathbb{F}_{p^m}+u\mathbb{F}_{p^m}$, where $n,s$ are positive integers with $\operatorname{gcd}(n,p)=1$, $\beta \in \mathbb{F}_{p^m}$, and $\alpha={\alpha_0}^{p^s} \in \mathbb{F}_{p^m}$ such that $x^n-\alpha_0$ is irreducible over $\mathbb{F}_{p^m}$. Moreover, we show that the non-trivial MDS symbol-pair $\alpha$-constacyclic codes of length $np^s$ over $\mathbb{F}_{p^m}$ only exist when $n=1,2$. Similarly, the non-trivial MDS symbol-pair $(\alpha+u\beta)$-constacyclic codes of length $np^s$ over $\mathbb{F}_{p^m}+u\mathbb{F}_{p^m}$ only exist when $\beta=0$ and $n=1,2$. Some of these MDS symbol-pair codes we present in this paper are new  and have relatively large pair distance. It is an interesting problem to consider the case that $x^n-\alpha_0$ is reducible over $\mathbb{F}_{p^m}$.

\section*{References}
\bibliographystyle{elsarticle-num}
\bibliography{Ref}

\begin{thebibliography}{10}
\expandafter\ifx\csname url\endcsname\relax
  \def\url#1{\texttt{#1}}\fi
\expandafter\ifx\csname urlprefix\endcsname\relax\def\urlprefix{URL }\fi
\expandafter\ifx\csname href\endcsname\relax
  \def\href#1#2{#2} \def\path#1{#1}\fi

\bibitem{Cassuto2011}
Y.~Cassuto, M.~Blaum, Codes for symbol-pair read channels, IEEE Trans. Inf.
  Theory 57~(12) (2011) 8011--8020.

\bibitem{Yaakobi2012}
E.~{Yaakobi}, J.~{Bruck}, P.~H. {Siegel}, Decoding of cyclic codes over
  symbol-pair read channels, in: IEEE Int. Symp. Inf. Theory, Cambridge, MA,
  USA, 2012, pp. 2891--2895.

\bibitem{Takita2015}
M.~Takita, M.~Hirotomo, M.~Morll, A decoding algorithm for cyclic codes over
  symbol-pair read channels, IEICE Trans. Fundam.Electron., Commun. Comput.
  Sci. E98.A~(12) (2015) 2415--2422.

\bibitem{Morii2016}
M.~{Morii}, M.~{Hirotomo}, M.~{Takita}, Error-trapping decoding for cyclic
  codes over symbol-pair read channels, in: Int. Symp. Inf.Theory Appl.,
  Monterey, CA, USA, 2016, pp. 681--685.

\bibitem{Chee2013}
Y.~M. Chee, L.~Ji, H.~M. Kiah, C.~Wang, J.~Yin, Maximum distance separable
  codes for symbol-pair read channels, IEEE Trans. Inf. Theory 59~(11) (2013)
  7259--7267.

\bibitem{Li2016}
S.~Li, G.~Ge, Constructions of maximum distance separable symbol-pair codes
  using cyclic and constacyclic codes, Des. Codes Cryptogr. 84~(3) (2016)
  359--372.

\bibitem{Chen2017}
B.~Chen, L.~Lin, H.~Liu, {Constacyclic Symbol-Pair Codes: Lower Bounds and
  Optimal Constructions}, IEEE Trans. Inf. Theory 63~(12) (2017) 7661--7666.

\bibitem{Kai2018}
X.~{Kai}, S.~{Zhu}, Y.~{Zhao}, H.~{Luo}, Z.~{Chen}, New {MDS} symbol-pair codes
  from repeated-root codes, IEEE Commun. Lett. 22~(3) (2018) 462--465.

\bibitem{Dinh2018}
H.~Q. Dinh, B.~T. Nguyen, A.~K. Singh, S.~Sriboonchitta, On the symbol-pair
  distance of repeated-root constacyclic codes of prime power lengths, IEEE
  Trans. Inf. Theory 64~(4) (2018) 2417--2430.

\bibitem{Dinh2019b}
H.~Q. {Dinh}, P.~{Kumam}, P.~{Kumar}, S.~{Satpati}, A.~K. {Singh}, W.~{Yamaka},
  {MDS} symbol-pair repeated-root constacylic codes of prime power lengths over
  $\mathbb{F}_{p^m}+ u \mathbb{F}_{p^m}$, IEEE Access 7 (2019) 145039--145048.

\bibitem{Dinh2020}
H.~Q. {Dinh}, B.~T. {Nguyen}, S.~{Sriboonchitta}, {MDS} symbol-pair cyclic
  codes of length $2p^s$ over $\mathbb{F}_{p^m}$, IEEE Trans. Inf. Theory
  66~(1) (2020) 240--262.

\bibitem{Zhao2020}
W.~Zhao, S.~Yang, K.~W. Shum, X.~Tang, Repeated-root constacyclic codes with
  pair-metric, IEEE Commun. Lett. (2020).
\newblock \href {https://doi.org/10.1109/lcomm.2020.3041271}
  {\path{doi:10.1109/lcomm.2020.3041271}}.

\bibitem{Dinh2018a}
H.~Q. Dinh, B.~T. Nguyen, A.~K. Singh, S.~Sriboonchitta, Hamming and
  symbol-pair distances of repeated-root constacyclic codes of prime power
  lengths over $\mathbb{F}_{p^{m}}+u \mathbb{F}_{p^{m}}$, IEEE Commun. Lett.
  22~(12) (2018) 2400--2403.

\bibitem{Dinh2019a}
H.~Q. Dinh, X.~Wang, H.~Liu, S.~Sriboonchitta, On the symbol-pair distances of
  repeated-root constacyclic codes of length $2 p^s$, Discrete Math. 342~(11)
  (2019) 3062 -- 3078.

\bibitem{Ozadam2009}
H.~{\"{O}}zadam, F.~{\"{O}}zbudak, Two generalizations on the minimum hamming
  distance of repeated-root constacyclic codes (2009).
\newblock \href {http://arxiv.org/abs/0906.4008} {\path{arXiv:0906.4008}}.

\bibitem{Sharma2019}
A.~Sharma, T.~Sidana, On the structure and distances of repeated-root
  constacyclic codes of prime power lengths over finite commutative chain
  rings, {IEEE} Transactions on Information Theory 65~(2) (2019) 1072--1084.
\newblock \href {https://doi.org/10.1109/tit.2018.2864293}
  {\path{doi:10.1109/tit.2018.2864293}}.

\bibitem{Huffman2003}
W.~C. Huffman, V.~Pless, Fundamentals of Error-Correcting Codes, Cambridge
  University Press, 2003.

\bibitem{RudolfLidl2008}
R.~Lidl, H.~Niederreiter, Finite Fields, Cambridge University Press, 2008.

\bibitem{Massey1973}
J.~Massey, J.~Costello, D.J., J.~Justesen, Polynomial weights and code
  constructions, IEEE Trans. Inf. Theory 19~(1) (1973) 101--110.

\bibitem{Dinh2007}
H.~Q. {Dinh}, Complete distances of all negacyclic codes of length $2^{s}$ over
  $\mathbb{Z}_{2^{a}}$, IEEE Trans. Inf. Theory 53~(1) (2007) 147--161.

\bibitem{Zhao2018}
W.~Zhao, X.~Tang, Z.~Gu, All $\alpha+u \beta$-constacyclic codes of length $n
  p^{s}$ over $\mathbb{F}_{p^{m}}+u \mathbb{F}_{p^{m}}$, Finite Fields Appl. 50
  (2018) 1--16.

\bibitem{Cao2015}
Y.~Cao, Y.~Cao, J.~Gao, F.~Fu, {Constacyclic codes of length $p^s n$ over
  $\mathbb{F}_{p^{m}}+u \mathbb{F}_{p^{m}}$} (2015).
\newblock \href {http://arxiv.org/abs/1512.01406v2}
  {\path{arXiv:1512.01406v2}}.

\end{thebibliography}

\end{document}